\documentclass[journal]{IEEEtran}
\usepackage{t1enc}
\usepackage{alltt}
\usepackage{epsfig}
\usepackage{graphicx,ifthen}
\usepackage{eurosym}
\usepackage{amsfonts}
\usepackage{latexsym}
\usepackage{amssymb,amsthm,amsmath}
\usepackage{verbatim}
\usepackage{multirow}
\usepackage{bm}
\usepackage{float}
\usepackage{algorithm}
\usepackage{algorithmic}
\usepackage{slashbox}
\begin{document}

\title{Queueing Game For Spectrum Access in Cognitive
Radio Networks}
\author{Zheng~Chang,~\IEEEmembership{Member,~IEEE,}
 Tapani~Ristaniemi,~\IEEEmembership{Senior~Member,~IEEE,}
Zhu~Han,~\IEEEmembership{Fellow,~IEEE,}
\thanks{Z. Chang and T. Ristaniemi are with University of Jyv\"akyl\"a, Department of Mathematical Information Technology, P.O.Box 35, FIN-40014 Jyv\"akyl\"a, Finland.
Z. Han is with the Electrical and Computer Engineering Department,
University of Houston, Houston, USA.
E-mail: zheng.chang@jyu.fi, zhan2@uh.edu, tapani.ristaniemi@jyu.fi}
}

\maketitle

\begin{abstract}
In this paper, we investigate the problem of spectrum access
decision-making for the Secondary Users (SUs) in the cognitive radio
networks. When the Primary Users (PUs) are absent on certain frequency
bandwidth, SUs can formulate a queue and wait for the Base Station
(BS) to serve. The queue of the SUs will be dismissed if the PU is
emerging in the system. Leveraging the queueing game approaches,
the decision-making process of the SUs that whether to queue or not is
studied. Both individual equilibrium and social optimization
strategies are derived analytically. Moreover, the optimal pricing
strategy of the service provider is investigated as well. Our
proposed algorithms and corresponding analysis are validated
through simulation studies.
\end{abstract}

\begin{IEEEkeywords}
queueing game; pricing; Nash equilibrium strategies; social
optimizations; spectrum access; cognitive radio
\end{IEEEkeywords}

\section{Introduction}
\label{Sec1}

In the fast growing wireless market, the frequency spectrum is one
of the most scarce and valuable resources. However, some surveys
for actual measurements show that most of the allocated spectrum is largely under-utilized \cite{Hossain}. Cognitive radio (CR)
is known as an efficient way to improve spectrum utilization and a
promising technology to enable dynamic spectrum access by
exploiting the unused spectrum in the wireless environments
\cite{Hossain}. In a cognitive radio network (CRN), there are two
types of users, namely, licensed primary users (PUs) who have the licensed spectrum access opportunities and unlicensed
secondary users (SUs) who can only utilize the spectrum which the PUs do not occupy. To exploit limited spectrum efficiently, In a CRN, the SUs are allowed to opportunistically to access licensed
spectrum bands when PU transmission is not presented, which is
able to significantly improve the spectrum utilization efficiency.
\par

During last decades, the technology of CR has attracted plenty of interests and has been investigated
extensively, among which, how to explore the access opportunities
and regular the spectrum access without harming the PUs are of
consistent research interests \cite{Zhao}. Recently, there is a
particular interest to apply the queueing theory \cite{Shiang} which is a natural tool to
analyze the transmission in the wireless networks or the
game theory \cite{Han} which is commonly used in developing optimization algorithms to address the spectrum access problems.
Intriguingly, applying queue theory with pricing strategies also
brings a novel view on understanding the SU's behavior and
cognitive market policies, which can be traced from the original
works of \cite{Naor}-\cite{Altman} that study the equilibrium
behavior in queueing systems. We call it the queueing game in the
following to highlight the key issues in this area: by utilizing the information of queue
and considering its own payoff, the user needs to make the decision on whether to queue or not.
Recent works of \cite{HLi}-\cite{Tran} have leveraged the
queueing game to the spectrum access control in the CRN.
In \cite{HLi}, the authors present an observed queue model with SUs
acquiring for transmission from a cognitive base station (CBS). It
is also assumed that PUs emerged periodically for transmission
opportunities. During the PU's transmission period, CBS stops
serving SUs and SUs remained in the queue waiting for the
re-operation of CBS. Based on this model, a SU's decision
strategy, i.e., joining the queue or balk, is investigated and
the pricing policy was studied. Authors of \cite{Do1} extend the work
of \cite{HLi} to an unobserved case and analyzed the strategies of
SUs and CBS. In \cite{Tran}, the strategy of delay sensitive SUs
is explicitly considered. The authors also consider
pricing and load balancing effect in the spectrum access
decision-making in both monopoly and duopoly markets. Authors of
\cite{Guan} also utilizes the concept of the queueing game and model
users as selfish players that compete with each other by choosing the
optimal transmission threshold maximizing system throughput.
\par

It can be well observed that the queueing game is an effective tool
for analyzing the SUs' behaviors in spectrum access. The motivation
of this work is to extend the previous works and overcome their
limitations. It can be found that the previous works focused on
the case that a separate CBS is used for serving the SUs. Therefore, when the PU
emerges, the SUs can remain in the queue of the CBS and wait for PUs finishing
the transmission. In such case, the queue of the CBS can be modeled and
analyzed by a server breakdown model. However, the installation and deployment of
the CBS bring additional cost and it may not be practical in some
cognitive systems. In contrast, we consider there is only one BS
in the system. The primary job of the BS is to serve the incoming PUs,
while it will utilize the time when the PU frequency band is not occupied to
serve SUs. Specially, when the PU comes, the BS no long holds the
information of the queue and all the SUs are forced to leave the queue.
When a SU decides to join the queue and wait for spectrum access opportunity, the sojourn time induces a
cost and if its job has been finished, the SU can receive a reward. In
addition, if PU emerges, the SUs who have to leave the queue
without any reward or compensate. Under this model, we study the
SU's decision-making process, i.e., whether to join for spectrum access or balk, and
the pricing based spectrum access control. \par

Compared to previous works, the main contributions of this paper
are as follows. We first model the interaction among SUs in a
partially observed queue as a noncooperative game. Based on the
queueing analysis and payoff model, we then analyze both
individually equilibrium strategy and optimal social welfare
strategy of the SUs about whether to join or not. Furthermore, we
study the BS pricing strategy for the system such that the
individually equilibrium decision of SUs can coincide with the
socially optimal strategy that optimizes the total welfare of the
whole system. Our presented analysis and algorithm are
demonstrated by the simulation studies. \par

The rest of this paper is organized as follows. Section
\uppercase\expandafter{\romannumeral2} describes the system model
including the queue model and profit of SUs. We present the individual
equilibrium strategies, social welfare optimization and pricing
studies in Section \uppercase\expandafter{\romannumeral3}. Our presented algorithm and analysis are demonstrated in Section
\uppercase\expandafter{\romannumeral4} through simulation studies,
and finally we conclude this work in Section
\uppercase\expandafter{\romannumeral5}.

\section{System Model}
\label{Sec2}

\subsection{Queue Model}
The system model can be found in Fig. \ref{fig:1}. We consider a
CRN consists of multiple SUs, one BS and multiple PUs. When PUs are
absent, the BS can utilize the available spectrum to serve the SUs. If
the PU is entering the system, due to its priority, the BS has to drop the
service connection of the SUs and starts to serve PU. Meanwhile, the
queue consisting of the SUs is no long exists and all SUs should leave the queue
and seek for other transmission opportunities. We consider the
data arrival rate of the SU follows the Poisson process at rate $\lambda$
and the arrival of the PU follows Poisson process at rate $\xi$. The
service requirements of the SUs are i.i.d with exponential distribution
$\mu$. The First-Come-First-Served (FCFS) rule is applied for
determining the service order of SUs at the BS. The length of the service
time of the PU is also assumed to be exponentially distributed with
rate $\eta$. The above considered queue model is a common assumption in
some previous literatures, e.g. \cite{HLi} \cite{Do1}. \par

The state of the system at time $t$ is represented by a pair
$(N(t), I(t))$, where $N(t)$ is the length of queue, i.e., the
number of SUs in the system. $I(t)$ denotes the working status of
the BS, with $1$ standing for serving SUs and $0$ showing that the BS
is serving the PUs. So based on the system model, when $I(t) = 0$, we
also have $N(t)=0$. We assume when being successfully served by the BS, the SU can receive a
reward. While waiting in the queue, the delay cost of the SU is a function of its sojourn time. Based on
the total payoff which is the difference between the reward and the cost, SU can make the decision on
whether to join the queue and waiting for spectrum access or not. The payoff models are presented in the following.

\begin{figure}[t]
\centering
\includegraphics[height=4 cm, width=6cm]{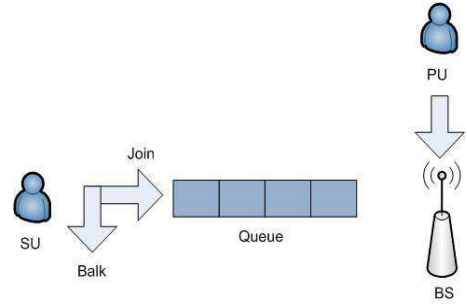}
\caption{System Model} \label{fig:1}
\end{figure}

\subsection{Profit Model}
We assume when being successfully served by the BS, the SU receives a
reward $\varpi_s$. We also assume that the cost of staying in the
queue is $\chi(T)$ where $T$ is the sojourn time in the system
representing the time that SU stays in the queue. $\chi(T)$ should be an
increasing function of $T$ and one simple linear example is that
$\chi(T) = C T$ where $C$ is the unit cost. Then the profit of the
SU can be given as

\begin{equation}
\label{eq:U} R := \varpi_{s} - CT.
\end{equation}

\noindent Accordingly, we can have the definition of the expected individual
profit

\theoremstyle{definition}
\newtheorem{definition}{Definition}
\begin{definition} When there are $n$ SUs in the queue, the expected profit of the arriving SU can be defined as
\begin{equation}
\label{eq:IP} U(n) :=  \theta_s (\varpi_s - C E(Q_n)),
\end{equation}
\noindent where $\theta_s$ is the probability that the SU can be
served and $E(Q_n)$ is the expected sojourn time related to $n$.
\end{definition}

\noindent Moreover, we can also define the social welfare as follows.

\begin{definition}
The social welfare of the considered system is given by
\begin{equation}
\label{eq:SPSO} S(q) := \lambda \rho_{s} \varpi_s - C E(N),
\end{equation}
\noindent where $\rho_{s}$ is the fraction of SUs that join the
queue and leave after being served. $E(N)$ is the mean number of
SUs.
\end{definition}

In the following, we consider the SUs are risk natural and try to
maximize their profits. We also assume the SUs are identical i.e.,
a symmetric game is assumed. In addition, we also consider that
when the SU enters the system at time $t$, all the system
parameters including profit model are known except the queue
length $N(t)$ and SU can only observe $I(t)$ upon arrival, i.e., a
partially observed queue is considered.

\section{Queueing Game for Spectrum Access: Equilibrium and Pricing}

\subsection{Stationary Probability and Expected Sojourn Time}

\begin{figure}[t]
\centering
\includegraphics[height=3cm, width=6cm]{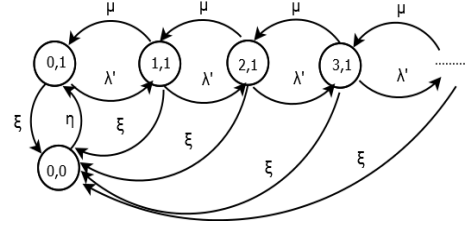}
\caption{Transition-rate diagram} \label{fig:2}
\end{figure}

After entering the system, the SU is able to decide whether to
join the queue waiting for spectrum access or not based on its own
observation, i.e. the profit. Thus, there are only two pure
strategies, join and balk, and a mixed strategy is specified by
the joining probability $q \in [0,1]$ of a SU that finds the BS is
not serving PU. Therefore, the effective arrival rate is
$\lambda^{'} = \lambda q$. For the considered system, we first
obtain the stationary probabilities $p(N(t), I(t))$. Based on the
transition rate diagram in Fig. \ref{fig:2}, we are able to obtain
the stationary probabilities $p(N(t), I(t))$ based on the
following balance equations,
\begin{equation}
\label{eq:SP1} \eta p(0,0)  =  \xi \sum_{n=0}^{+ \infty} p(n,1),
\end{equation}

\begin{equation}
\label{eq:SP1.1} (\lambda^{'} + \xi) p(0,1) =  \mu p(1,1) + \eta p
(0,0),
\end{equation}

\begin{equation}
\label{eq:SP1.2}  (\lambda^{'} + \xi + \mu) p(n,1) =  \lambda^{'}
p(n-1,1) + \mu p (n+1,1), \forall n \geq 1,
\end{equation}

\begin{equation}
\label{eq:SP1.3} p(0,0) + \sum_{n=0}^{+ \infty} p(n,1) = 1.
\end{equation}

To obtain $p(0,0)$ and $p(n,1)$, we have following observation.

\newtheorem{proposition}{Proposition}
\begin{proposition}
The stationary probability $p(0,0)$ and $p(n,1)$ can be given as
follows,

\begin{equation}
\begin{split}
\label{eq:SP11}
p(0,0) & =  \frac{\xi}{\eta+\xi}, \\
p(n,1) & =  \frac{\eta(1-x(\lambda^{'} ))x(\lambda^{'}
)^n}{\eta+\xi}
\end{split}
\end{equation}

\noindent where $x(\lambda^{'} )$ is

\begin{equation}
\begin{split}
\label{eq:SP2}
x(\lambda^{'} ) =  \frac{(\lambda^{'} + \mu + \xi)- \sqrt{(\lambda^{'} + \mu + \xi)^2 - 4 \lambda^{'}  \mu}}{2\mu}. \\
\end{split}
\end{equation}
\end{proposition}

\begin{proof}
The proof is given in Appendix A.
\end{proof}

In order to achieve the expected profit, from (\ref{eq:IP}),
$\theta_s$ and $E(Q_n)$ should be obtained at the first place.
Consider a SU enters the network at state $(n, 1)$ upon arrival
and decides to join the queue. This SU may leave either due to
service completion or due to a PU entering. For its service
completion, the SU has to wait for a sum of $n + 1$ independent
exponentially distributed times with parameter $\mu$. For the case
that the PU enters, the SU has to wait for an exponentially distributed
time with parameter $\xi$. Hence, the sojourn time of the SU in the
system is given as $N = \min(L_n, Q)$, where $L_n$ follows a
Gamma distribution with parameters $n+1$ and $\mu$. $Q$ is an
exponentially distributed random variable with rate $\xi$, and $Q$
is independent of $L_n$. Therefore, we have $\theta_s=Pr(L_n <
Q)$. To this end, when considering $N(t)$ and $I(t)$ are known to the SU,  we can obtain the $\theta_s$ and $E(Q_n)$ in
(\ref{eq:IP}) as

\begin{equation}
\begin{split}
\label{eq:A1} \theta_s & = \left(\frac{\mu}{\mu +
\xi}\right)^{n+1}, \\
E(Q_n) & = \frac{1}{\xi}\left( 1 - \theta_s \right)^{n+1},
\end{split}
\end{equation}

\noindent where $n$ is the number of SUs in the queue. Then, we
can use (\ref{eq:SP1}) and (\ref{eq:A1}) to address $U(n)$ in
(\ref{eq:IP}) and then find the expected profit of a SU that enters
the queue with a certain probability.

\subsection{Individually Equilibrium Strategy}

With the results of the stationary probability and the expected
sojourn time, the expected profit of a SU that enters the queue
with probability $\tilde{q}$ can be achieved as following,

\begin{proposition}
When there is frequency bandwidth available for the SUs, the expected
profit of a SU that enters the queue with probability $\tilde{q}$
given that other SUs join the queue with probability $q$ is given
by
\begin{equation}
\label{eq:ProfitExpected}
\Gamma(\tilde{q},q) = \tilde{q}
\left[\left( \varpi_s + \frac{C}{\xi} \right)
\frac{\mu(1-x(\lambda^{'}))}{\mu+\xi-\mu x(\lambda^{'})}-
\frac{C}{\xi}\right]
\end{equation}
\end{proposition}

\begin{proof}
The proof is given in Appendix B.
\end{proof}

We can now process to find out the individually
equilibrium (IE) strategy of a SU and we have the following,

\theoremstyle{theorem}
\newtheorem{theorem}{Theorem}
\begin{theorem}
In the considered model, a unique IE mixed strategy exists, with
joining probability $q_e$ given by

\begin{equation}
\label{eq:NE}
 q_{e}=
\begin{cases}
1, & \varpi_s \in [\frac{C}{\mu(1-\kappa )},+\infty), \\
\Theta, &\varpi_s \in (\frac{C}{\mu},\frac{C}{\mu(1-\kappa )}),\\
0, &\text{Otherwise.}\\
\end{cases}
\end{equation}

\noindent where $\Theta = \frac{(\mu \varpi_s -C)(C + \xi
\varpi_s)}{\lambda C \varpi_s }$ and $\kappa = x(\lambda)$.

\end{theorem}
\begin{proof}
The proof is given in Appendix C.
\end{proof}

From \textbf{Theorem 1}, we can observe that the IE is independent
of the PU's service time. This is because the SU makes decision
only when PU is not being served and the SU can not be aware of
the PU's information. However, the socially equilibrium strategies
of the SUs do depend on the pricing strategies of the BS as well
as the PU's arrival rate. The IE does not imply the social welfare
optimality. Thus, in the following, we investigate the social
optimal strategy.

\subsection{Social Optimization}

Based on the definition of social welfare and previous results of
the stationary probabilities, we can arrive at the following
proposition,

\begin{proposition}
The expected social profit, given that the SUs follow a mixed
strategy with probability $q$ of joining (i.e. arriving
SUs that find an BS not serving the PU enters with probability $q$,
while the rest choose to balk without being served) is given by

\begin{equation}
\label{eq:ProfitSO} S(q) = \frac{\eta x(\lambda^{'}) [\mu \varpi_s
(1-x(\lambda^{'})) - C]}{(\xi + \eta)(1-x(\lambda^{'}))}.
\end{equation}

\end{proposition}

\begin{proof}
To address the social welfare (\ref{eq:SPSO}), we need to find
$E(N)$ and $\rho_s$. Using (\ref{eq:SP1}) and (\ref{eq:A1}), we
can obtain $\rho_s$ and $E(N)$ as

\begin{equation}
\begin{split}
\label{eq:C1} &\rho_s  = \sum_n p(n,1) q \left(\frac{\mu}{\mu +
\xi}\right)^{(n+1)}, \\
& E(N)  = \sum_n n p(n,1). \\
\end{split}
\end{equation}

Then, obtaining the geometric sums in (\ref{eq:C1}) and through
(\ref{eq:SPSO}) we can arrive (\ref{eq:ProfitSO}).
\end{proof}

With the observation in (\ref{eq:ProfitSO}), we can obtain the
socially optimal strategy, which can be found in the following
theorem.

\begin{theorem}
In the considered model, a unique socially optimal strategy exists
with probability $q_{s}$ of joining the queue which can be
expressed as
\begin{equation}
\label{eq:SOeq}
 q_{s}=
\begin{cases}
1, & \varpi_s \in [\frac{C}{\mu(1-\kappa )^2},+\infty), \\
\Phi, &\varpi_s \in (\frac{C}{\mu},\frac{C}{\mu(1-\kappa)^2}),\\
0, &\text{Otherwise.}\\
\end{cases}
\end{equation}

\noindent where $\Phi = \frac{\sqrt{\vartheta}(\mu \varpi_s -
\sqrt{\vartheta})(\xi \varpi_s + \sqrt{\vartheta})}{\lambda
\vartheta \varpi_s}$ and $\vartheta = \mu \varpi_s C$.

\end{theorem}
\begin{proof}
The proof is given in Appendix D.
\end{proof}

\subsection{Optimal Pricing}

We have obtained the socially optimal strategy as well as the
individual equilibrium strategy. Moreover, it can be observed that
the socially optimal joining probability $q_s$ is always smaller
than the individual one $q_e$, which can also be found in Fig. \ref{fig:3}. To oblige the SUs to adopt the
socially optimal strategy, one approach is to apply a pricing
mechanism to reduce the individually optimal threshold $q_e$
\cite{HLi}. In this work, we consider that the BS will act as an
agent to impose an admission fee, which is a constant given the
arrival rate, service pattern, reward and cost. The admission fee
is to force the individually equilibrium probability to equal with
the social optimal one. \par

When the admission fee is considered, the expected profit of a SU
is given as $U(n,p) = \theta_s (\varpi_s - C E(Q_n)-p)$. It can be
observed that when imposing an admission fee, the social profit
remains the same as it implies a transfer of income from one group
to another. Thus, through \ref{eq:A2}, we can obtain $\Gamma(q_s,
q_s, p)$. We further denote that the equilibrium probability of
joining by $q_e(p)$. Then the optimal fee $p^\ast$ should satisfy
$q_s = q_e(p^\ast)$. As the monopoly market is considered here,
and a monpoly does not allow the a positive user surplus since in
such a case, the price can be increased without reducing $q$.
Therefore, the $p^{\ast}$ can be arrived by

\begin{equation}
 p^{\ast}= \{p \vert \Gamma_p(q_s, q_s,p) = 0\},
\end{equation}

\section{Simulation Results}
The numerical parameters are $\lambda = 7, \xi = 0.5, \mu = 3,
\eta=2, C = 2$. First, in Fig. \ref{fig:3}, we can see that the
socially optimal joining probability $q_s$ is always smaller than
the individual one $q_e$, and inherently, there is a gap between
the individually equilibrium arrival rate and social arrival rate
as the equilibrium arrival rate and social arrival rate can be
given as $\lambda_{e/s} = \lambda q_{e/s}$. Therefore, we can see
that by imposing appropriate admission fee, the arrival rate of
SUs can be regulated. In Fig. \ref{fig:4}, we can see that the
social benefit when $q_s$ is used is higher than the one when
$q_e$ is considered. Comparing Fig. \ref{fig:3} and Fig.
\ref{fig:4}, we can see that only when $q_e=1$, the social benefit
is higher than $0$, otherwise it remains lower or equal to $0$.
Meanwhile, the social benefits when $q_s$ is considered is always
above $0$ and comparable high, which indicates that appropriate
admission fee can improve the social benefits of the considered
system.

\begin{figure}[t]
\centering
\includegraphics[height=6cm, width=8cm]{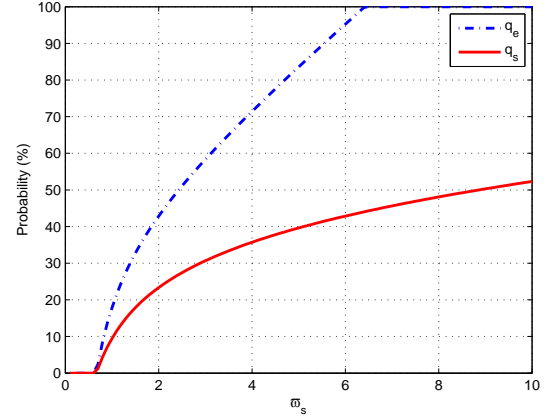}
\caption{The equilibrium joining probability vs. reward} \label{fig:3}
\end{figure}

\begin{figure}[t]
\centering
\includegraphics[height=6cm, width=8cm]{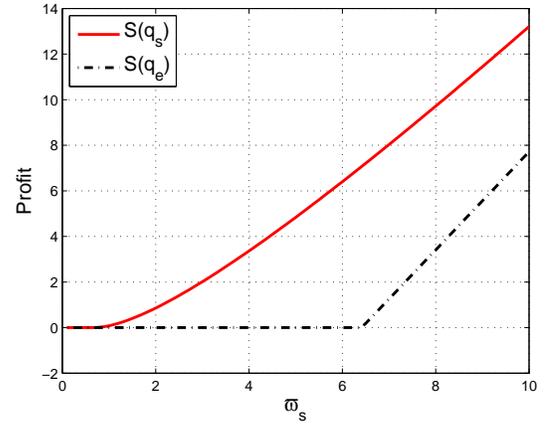}
\caption{The profits vs. reward} \label{fig:4}
\end{figure}

%

\section{Conclusion}
In this paper, we have investigated the problem of spectrum access
decision-making for the SUs in the CRNs. Utilizing
the queueing game approaches, the decision-making process of the SUs that
whether to queue or not is studied in present of arrival of the PUs.
Both individual equilibrium and social optimization strategies are
derived. Moreover, the optimal pricing of the service provider is also
investigated. Our proposed algorithms and corresponding analysis
are validated through simulation studies.

\section*{Appendix A}

From (\ref{eq:SP1})-(\ref{eq:SP1.3}), we can observe
$p(0,0) = \frac{\xi}{\eta+\xi}$. To obtain $p(n,1)$, the similar
approach used in \cite[pp. 578]{Boudali} can be applied. One can notice
that (\ref{eq:SP1.2}) can be considered as a homogeneous linear
difference equation of order $2$, which is with constant
coefficients and characteristic equation

\begin{equation}
\label{eq:O1} (\lambda^{'} + \mu + \xi)y = \lambda^{'}+ \mu y^2.
\end{equation}

\noindent (\ref{eq:O1}) have too roots, $x^{'}(\lambda^{'})$ and
$x(\lambda^{'})$, which are

\begin{equation}
\begin{split}
\label{eq:O2}
x^{'}(\lambda^{'} ) & =  \frac{(\lambda^{'} + \mu + \xi) + \sqrt{(\lambda^{'} + \mu + \xi)^2 - 4 \lambda^{'}  \mu}}{2\mu}, \\
x(\lambda^{'} ) & =  \frac{(\lambda^{'} + \mu + \xi) -\sqrt{(\lambda^{'} + \mu + \xi)^2 - 4 \lambda^{'}  \mu}}{2\mu}. \\
\end{split}
\end{equation}

From the standard theory of homogeneous linear difference
equations (see e.g. \cite[Sec 2.3]{Elaydi}), we conclude that

\begin{equation}
\label{eq:O3} p(n,1) = c^{'}x^{'}(\lambda^{'} )^n + c
x(\lambda^{'} )^n,
\end{equation}

\noindent where $c^{'}$ and $c$ are constants. We can easily see
 that $x^{'} > 1$, thus, $c^{'}$ should be necessarily $0$. The constant $c$ can be
calculated using (\ref{eq:SP1.3}) and \textbf{Proposition 1} can
be approved.

\section*{Appendix B}

To prove Proposition 2, first we have the expected profit when a
SU is able to observe $n$ SUs and queue state in the system upon
arrival (i.e. observed queue) and decide to enter, which can be
derived from (\ref{eq:A1}),

\begin{equation}
\label{eq:A3} U(n)= \varpi_s\left( \frac{\mu}{\mu+\xi}
\right)^{(n+1)} - \frac{C}{\xi} + \frac{C}{\xi}\left(
\frac{\mu}{\mu+\xi} \right)^{(n+1)}.
\end{equation}

When a SU decides to join given that the BS is serving SUs and
others are following strategy $q$, the expected profit is

\begin{equation}
\begin{split}
\label{eq:A2}
\Gamma(1,q) &= \sum_{n=0}^{+\ \infty}p_{1}(n,1) U(n) \\
&= \sum_{n=0}^{+ \infty} \frac{p(n,1)}{\sum_{i=0}^{+ \infty}p(i,1)} U(n) \\
&= \left( \varpi_s+ \frac{C}{\xi} \right)
\frac{\mu(1-x(\lambda^{'}))}{\mu+\xi-\mu x(\lambda^{'})} - \frac{C}{\xi}, \\
\end{split}
\end{equation}

\noindent where $p_{1}(n,1) = \frac{p(n,1)}{\sum_{i=0}^{+
\infty}p(i,1)}$ is the probability that there are $n$ SUs in the
queue when a SU arrives. As we have that $\Gamma(\tilde{q},q) =
(1-\tilde{q})\Gamma(0,q)+ \tilde{q}\Gamma(1,q)$,
(\ref{eq:ProfitExpected}) can be arrived and \textbf{Proposition
2} can be proved.

\section*{Appendix C}

For a SU, it will prefer to enter the queue if its expected profit
after entering $\Gamma(1,q)>0$, and balk if it $\Gamma(1,q)<0$.
Considering $\Gamma(1,q)=0$, we can solve for $x(\lambda^{'})$ and
the unique solution is $x = 1- \frac{C}{\mu \varpi_s }$. \par

We notice that $x(\lambda^{'}), \lambda^{'} = \lambda q$ is
one root of (\ref{eq:O1}). Then we are able to solve
$(\ref{eq:O1})$ with respect to $q$, which yields

\begin{equation}
\label{eq:B2} q_e = \frac{x[\mu(1-x)+\xi]}{\lambda(1-x)}=
\frac{(\mu \varpi_s -C)(C + \xi \varpi_s)}{\lambda C \varpi_s }.
\end{equation}

We can also observe that $x(\lambda^{'})$ is a strictly
increasing function for $q \in [0,1]$ as $\frac{dx(\lambda^{'})}{d
q} >0$. Thus, $ q_e \in (0,1)$ iff $x(\lambda q_e) \in
(0,\kappa)$, where $\kappa = x(\lambda)$. In other word, $q_e$ is
in the interval $(0,1)$ iff $\varpi_s \in
(\frac{C}{\mu},\frac{C}{\mu(1-\kappa )})$. \par

Moreover, it can be found that when $\varpi_s \geq
\frac{C}{\mu(1-\kappa )}$, $\Gamma(1,q)$ keeps positive. Thus, the
SU's best response is to join. In this case, "join" is the unique
individual equilibrium. On the other hand, when $\varpi_s \leq
\frac{C}{\mu}$, $\Gamma(1,q)$ becomes negative. In this case,
"balk" is the individuals equilibrium strategy. Therefore,
\textbf{Theorem 1} can be proved.

\section*{Appendix D}

It can be noticed that (\ref{eq:ProfitSO}) can be reformed as
$S(q) = f(x(\lambda q))$, where $f(x) = \frac{\eta x[mu \varpi_s
(1-x)-C]}{(\xi + \eta )(1-x)}$. As we can see, $S^{'}(q) =0$ can
be deduced to $f^{'}(x(\lambda q))=0$, which means that we need to
solve $f'(x)=0$, that is

\begin{equation}
\label{eq:D1}
 \mu \varpi_s x^2 - 2 \mu \varpi_s x +\mu \varpi_s
-C=0.
\end{equation}

We can see that the discriminant of the quadratic polynomial in
(\ref{eq:D1}) is less or equal to $0$ iff $C \leq 0$. We can
deduce that $f(x)$ is increasing and so as $S(q)$. (\ref{eq:D1})
has two roots, which are
\begin{equation}
\label{eq:D2} x_1 = 1-\frac{\sqrt{\vartheta}}{\mu \varpi_s}, x_2 =
1+ \frac{\sqrt{\vartheta}}{\mu \varpi_s}.
\end{equation}

It is apparently $x_1<1<x_2$ and we have following discussions,

\begin{itemize}
    \item When $x_1\leq 0$, then we have $x_1<0<\kappa<x_2$. Then S(q)
    is decreasing in $[0,1]$ and the social optimal is $q_s=0$. We
    can also see that $x_1 \leq 0$ implies that $\varpi_s \leq \frac{C}{\mu}$.
    \item When $0<x_1<\kappa$, then we can see that the maximum
    $S(q)$ is obtained for $q$ such that $x(q) = x_1$. By using
    (\ref{eq:O1}), substituting $x_1$ for $y$ and solving for $q$,
    we obtain $q_s = \frac{\sqrt{\vartheta}(\mu \varpi_s -
\sqrt{\vartheta})(\xi \varpi_s + \sqrt{\vartheta})}{\lambda
\vartheta \varpi_s}$. It can be observed that $0<x_1<\kappa$ means
$\frac{C}{\mu}< \varpi_s < \frac{C}{\mu (1-\kappa)^2}$.

    \item Similarly, when $x_1 \geq \kappa$ we can have the social
    optimal joining probablity is $q_s = 1$ and $\varpi_s \geq \frac{C}{\mu (1-\kappa)^2}$.
\end{itemize}


\begin{thebibliography}{1}

\bibitem{Hossain}
E. Hossain, D. Niyato, and Z. Han, \emph{Dynamic Spectrum Access and
Management in Cognitive Radio Networks}. Cambridge University Press
Cambridge, 2009.

\bibitem{Zhao}
Q. Zhao, and B. M. Sadler, "A survey of dynamic spectrum access,"
\emph{IEEE Signal Processing Magazine}, vol. 24, no. 3, pp. 79-89, May
2007.

\bibitem{Shiang}
H.-P. Shiang and M. van der Schaar, "Queuing-based dynamic channel
selection for heterogeneous multimedia applications over cognitive
radio networks," \emph{IEEE Trans. Multimedia}, vol. 10, no. 5,
pp. 896-909, Aug. 2008.

\bibitem{Yang}
L. Yang, H. Kim, J. Zhang, M. Chiang, and C. W. Tan,
"Pricing-based decentralized spectrum access control in cognitive
radio networks," \emph{IEEE/ACM Trans. Networking}, vol.
21, no. 2, pp. 522-535, Apr. 2013.

\bibitem{Han}
Z. Han, D. Niyato, W. Saad, T. Basar and A. Hjoungnes, \emph{Game Theory in
Wireless and Communication Networks: Theory, Models, and
Applications}. Cambridge University Press, 2011.

\bibitem{Naor}
P. Naor, "The regulation of queue size by levying tolls,"
\emph{Econometrica}, vol. 37, no. 1, pp. 15-24,Jan. 1969.

\bibitem{Hassin}
R. Hassin, and M. Haviv, \emph{To queue or not to queue: Equilibrium
Behavior in Queueing Systems}, Springer, New York, 2003.

\bibitem{Boudali}
O. Boudali, and A. Economou, "The effect of catastrophes on the
strategic customer behavior in queueing systems," \emph{Naval
Research Logistics}, vol 60, no. 7, pp. 571-587, Oct. 2013.


\bibitem{Altman}
E. Altman, "Applications of dynamic games in queues," in
\emph{Advances in Dynamic Games}. Boston, MA: Birkhauser, 2005.


\bibitem{HLi}
H. Li and Z. Han, "Socially optimal queuing control in CR networks
subject to service interruptions: to queue or not to queue?"
\emph{IEEE Trans. Wireless Commun.}, vol. 10, no. 5, pp. 1656-1666, May 2011.

\bibitem{Do1}
C. T. Do, N. H. Tran, M. V. Nguyen, C. S. Hong, and S. Lee,
"Social optimization strategy in unobserved queueing systems in
cognitive radio networks," \emph{IEEE Commun. Lett.}, vol. 16, no.
12, pp. 1944-1947, Dec. 2012.

\bibitem{Tran}
N. H. Tran, C. S. Hong, S. Lee, and Z. Han, "Optimal pricing
effect on equilibrium behaviors of delay-sensitive users in
cognitive radio networks," \emph{IEEE J. Sel. Areas Commun.}, vol.
31, no. 11, pp. 2566-2579, Nov. 2013.


\bibitem{Guan}
Z. Guan, T. Melodia, and G. Scutari, "Distributed queueing games in
interference-limited wireless networks," \emph{IEEE ICC},
Budapest, Hungary, Jun. 2013.


\bibitem{Elaydi}
S. Elaydi, \emph{An Introduction to Diffierence Equations}, 2nd Edition.
Springer, New York, 1999.




\end{thebibliography}
\end{document}